\documentclass[12pt]{iopart}
\usepackage{graphicx}

\usepackage{comment}
\usepackage{iopams}

\newcommand{\bea}{\begin{eqnarray}}
\newcommand{\eea}{\end{eqnarray}}
\newcommand{\beq}{\begin{equation}}
\newcommand{\eeq}{\end{equation}}
\newcommand{\nn}{\nonumber}
\newcommand{\eqref}[1]{(\ref{#1})}

\newcommand{\C}[1]{{\mathcal{#1}}}

\newcommand{\BB}[1]{{\mathbb{#1}}}

\newcommand{\half}{{\frac{1}{2}}}

\newcommand{\quarter}{{\frac{1}{4}}}

\newcommand{\threehalves}{{\frac{3}{2}}}

\newcommand{\abs}[1]{\vert #1\vert}
\newcommand{\avg}[1]{\left\langle #1 \right\rangle}

\newcommand{\res}{\mathrm{res}}
\newcommand{\sgn}{\mathrm{sgn}}

\newcommand\Tinv{\frac{1}{T}}

\newtheorem{theorem}{Theorem}
\newtheorem{lemma}[theorem]{Lemma}

\newenvironment{proof}[1][Proof]{\begin{trivlist}
\item[\hskip \labelsep {\bfseries #1}]}{\end{trivlist}}

\newcommand{\qed}{\nobreak \ifvmode \relax \else
      \ifdim\lastskip<1.5em \hskip-\lastskip
      \hskip1.5em plus0em minus0.5em \fi \nobreak
      \vrule height0.75em width0.5em depth0.25em\fi}

\begin{document}

\title{Continuum random combs and scale dependent  spectral dimension}
\author{Max R Atkin, Georgios Giasemidis and John F Wheater}
\address{Rudolf Peierls Centre for Theoretical Physics, 1 Keble Road, Oxford OX1 3NP, UK}
\ead{\mailto{m.atkin1@physics.ox.ac.uk}, \mailto{g.giasemidis1@physics.ox.ac.uk}, \mailto{j.wheater1@physics.ox.ac.uk}}

\begin{abstract} Numerical computations have suggested that in causal dynamical triangulation models of quantum gravity (CDT) the effective dimension of spacetime in the ultra-violet is lower than in the infra-red. In this paper we develop a simple model based on previous work on random combs, which share some of the properties of CDT, in which this effect can be shown to occur analytically. We construct a definition for short and long distance spectral dimensions and show that the random comb models exhibit scale dependent spectral dimension defined in this way. We also observe that a hierarchy of apparent spectral dimensions may be obtained in the cross-over region between UV and IR regimes for suitable choices of the continuum variables. Our main result is valid for a wide class of tooth length distributions thereby extending previous work on random combs by Durhuus et al. \end{abstract}

\submitto{\JPA}

\section{Introduction}
This work was motivated by observations made in some formulations of quantum gravity which we will explain shortly. However, the question as to whether it is possible to define consistently a spectral dimension which depends on the scale at which a system is probed by a random walk is of more general interest. In this paper we will construct a definition for such a spectral dimension and show that there are models which do indeed exhibit scale dependent spectral dimension defined in this way.

The demonstration by t' Hooft and Veltman  that quantised General Relativity is perturbatively non-renormalisable in four dimensions \cite{hooft}  led to the search for non-perturbative formulations of quantum gravity and there are now several approaches to this problem.
It was first advocated by Weinberg \cite{weinberg}  that there might be a non-trivial ultraviolet fixed-point and this has been pursued in continuum calculations by a number of authors \cite{Litim,Lauscher}. There is now quite a lot of evidence for such a fixed point although it is not conclusive; precisely because GR is perturbatively non-renormalisable it is necessary to make some assumptions or ultimately uncontrolled approximations in these calculations. An alternative approach within the fixed-point philosophy is to discretise space-time and to look for a critical point or line where a continuum limit may be taken to recover continuum gravity. 

Early attempts based on the so-called Euclidean quantum gravity model (see for instance \cite{Ambjorn:1997}) did not lead to a  continuum limit in four dimensions but the situation improved with the introduction of the  Causal Dynamical Triangulation model (CDT) by Ambj\o rn and Loll in 1998 \cite{Ambjorn:1998}. This defines the gravitational path integral as a sum over discretised geometries (see \cite{Ambjorn:2006} for a recent review). In contrast with the earlier  Euclidean quantum gravity model,  the CDT  approach takes  account of the Lorentzian nature of the path integral by building in a well defined temporal structure from the start. As an example of the success of the CDT approach numerical simulations have shown that in four-dimensional CDT large scale structure in the form of a four-dimensional de Sitter universe emerges purely from quantum fluctuations \cite{Ambjorn:2004aa,Ambjorn:2008aa}. This is a highly non-trivial result keeping in mind that one is dealing with a background independent formulation. Other approaches to quantum gravity include string theory, loop quantum gravity and causal sets but it is the results on the nature of space-time obtained using the CDT and fixed point calculations that particularly concern us here.

To discuss the nature of a quantum space-time at any distance scale requires a quantitative measure which is universal; that is to say it can be defined in any of the models we are interested in and is insensitive to cut-off scale physics while conveying information about the longer distance structure. The simplest such characterisations are various definitions of dimension of which the most familiar is the Hausdorff dimension.  The Hausdorff dimension $d_h$ is defined provided the volume $V(R)$ of a ball of radius $R$ takes the form 
\beq V(R)\sim R^{d_h}\eeq
if $R$ is large enough. 
An alternative probe of structure comes from the behaviour of unbiased random walks (equivalently diffusion) in whatever ensemble of space-times is being considered. The probability $P(t)$ that the walk returns to its starting point after time $t$ provides one of the simplest probes of the nature of space time in quantum gravity models. 
The spectral dimension $d_s$  is defined if
\beq P(t)\sim t^{-d_s/2}.\label{dsdef}\eeq
For which values of $t$ this should be true is a subtlety with which we will be concerned in this paper. For random walks on discretised graphs where the walker is allowed to hop from one vertex to a neighbour at each time step, behaviour of the form \eqref{dsdef} is expected at $t\to \infty$ i.e. when the walk is much longer than the short distance cut-off scale. On the other hand in the continuum the classic picture is that \eqref{dsdef} describes the behaviour as  $t\to 0$. Of course these are not, at least in principle, incompatible as they can be  related by a rescaling. 

An unexpected result from the numerical simulation of CDTs is that the  spectral dimension apparently varies from four at large scales to two at small scales  \cite{Ambjorn:2005aa}. Very recently similar results for this phenomenon of dimensional reduction have also been observed by other approaches such as the exact renormalisation group \cite{Litim,Lauscher},  Horava-Lifshitz gravity \cite{Horava:2009aa,Horava:2009ab}, and in three-dimensional CDT \cite{benedetti} and some further implications are discussed in \cite{Carlip:2009aa}. Such behaviour is not \emph{ a priori} implausible in quantum gravity because there is a dimensionful parameter, namely Newton's constant or equivalently the Planck length.
Studying the spectral dimension at different distance scales raises questions of definition and in the case of numerical simulations, discretisation problems. In a numerical simulation the largest available distance scale is determined by what will fit in the computer and short distance scales are often not much greater than the ultraviolet cut-off, or discretisation scale. Ideally there should be a hierarchy in which the long distance scale is much greater than the short distance scale which is in turn much greater than the cut-off. In this paper we develop a simple model based on previous work on random combs \cite{Durhuus:2005fq}. These are a family of simple geometrical models which share some of the properties of the CDT model; instead of an ensemble of triangulations we have an ensemble of graphs consisting of an infinite spine with teeth of identically independently distributed length hanging off (we define these graphs precisely in Section 2). It was shown in \cite{Durhuus:2005fq} that the spectral dimension is determined by the probability distribution for the length of the teeth. In this paper we show that it is possible to extend the work of \cite{Durhuus:2005fq} by taking a continuum limit thus ensuring that the cut-off scale is much shorter than all physical distance scales. We find that the spectral dimension is one if we take the physical distance explored by the random walk to zero and there exists a number of continuum limits in which the long distance spectral dimension differs from its short distance counterpart. As a by-product of this work we also extend some of the proofs given in \cite{Durhuus:2005fq} to a wider class of probability distributions.

This paper is organized as follows. In Section 2 we briefly review some known results for combs and their spectral dimension and then explain how in principle these can be extended to show different spectral dimensions at long and short distance scales. In Section 3 we introduce a simple model which we prove does in fact exhibit a spectral dimension that is different in the UV and IR. This model forms the basis of all later generalisations. In Section 4 we generalise the results of Section 3 to combs in which teeth of any length may appear with a probability governed by a power law. In Section 5 we examine the possibility of intermediate scales in which the spectral dimension differs from both its UV and IR values. In Section 6 we analyse the case of a comb in which the tooth lengths are controlled by an arbitrary probability distribution and show that continuum limits exist in which the short distance spectral dimension is one while the long distance spectral dimension can assume values in one-to-one correspondence with the positions of the real poles of the Dirichlet series generating function for the probability distribution. We then show how these techniques can be used to extend the results of \cite{Durhuus:2005fq}. In Section 7 we discuss our results and possible directions for future work. Some technical matters are contained in the appendices.


\section{Combs and Walks}

In this section we review some basic facts about random combs and random walks. As much as possible we use the same notation and conventions as \cite{Durhuus:2005fq} and refer to that paper for proofs omitted here. 

\subsection{Definitions}

We use the definition of a comb given in \cite{Durhuus:2005fq}.  Consider the nonnegative integers regarded as a graph, which we denote $N_\infty$, so that $n$
has the neighbours $n\pm 1$ except for $0$ which only has $1$ as a neighbour.
Furthermore, let $N_\ell$ be the integers $0,1,\ldots ,\ell$ regarded as a graph so that
each integer $n\in N_\ell$ has two neighbours $n\pm 1$ except for $0$ and $\ell$
which only have one neighbour, $1$ and $\ell-1$, respectively. A comb $C$ is an infinite rooted tree-graph with a special subgraph $S$ called the spine which is isomorphic to $N_\infty$ with the root at $0$.  At each vertex of
$S$, except the root $0$,  there is attached one of the graphs $N_\ell$ or
$N_\infty$.  We adopt the convention that these linear graphs which are
glued to the spine are attached at their endpoint $0$.  The linear graphs
attached to the spine are called the teeth of the comb, see figure \ref{fig1}.
We will find it convenient to say that a vertex on the spine with no tooth
has a tooth of length $0$.  We will denote by $T_n$ the tooth attached to
the vertex $n$ on $S$, and by $C_k$ the comb obtained by removing the links 
$(0,1),\ldots ,(k-1,k)$, the teeth $T_1,\ldots ,T_k$ and relabelling the 
remaining vertices on the spine in the obvious way. The number of nearest neighbours of a vertex $v$ will be denoted $\sigma(v)$.

It is convenient to give names to some special combs which occur frequently. We denote by $C=*$  the full comb in which every vertex on the spine is attached to an infinite tooth, and by $C=\infty$  the empty comb in which the spine has no teeth (so an infinite tooth is itself an example of $C=\infty$).

\begin{figure}[t]
  \begin{center}
    \includegraphics[width=10cm]{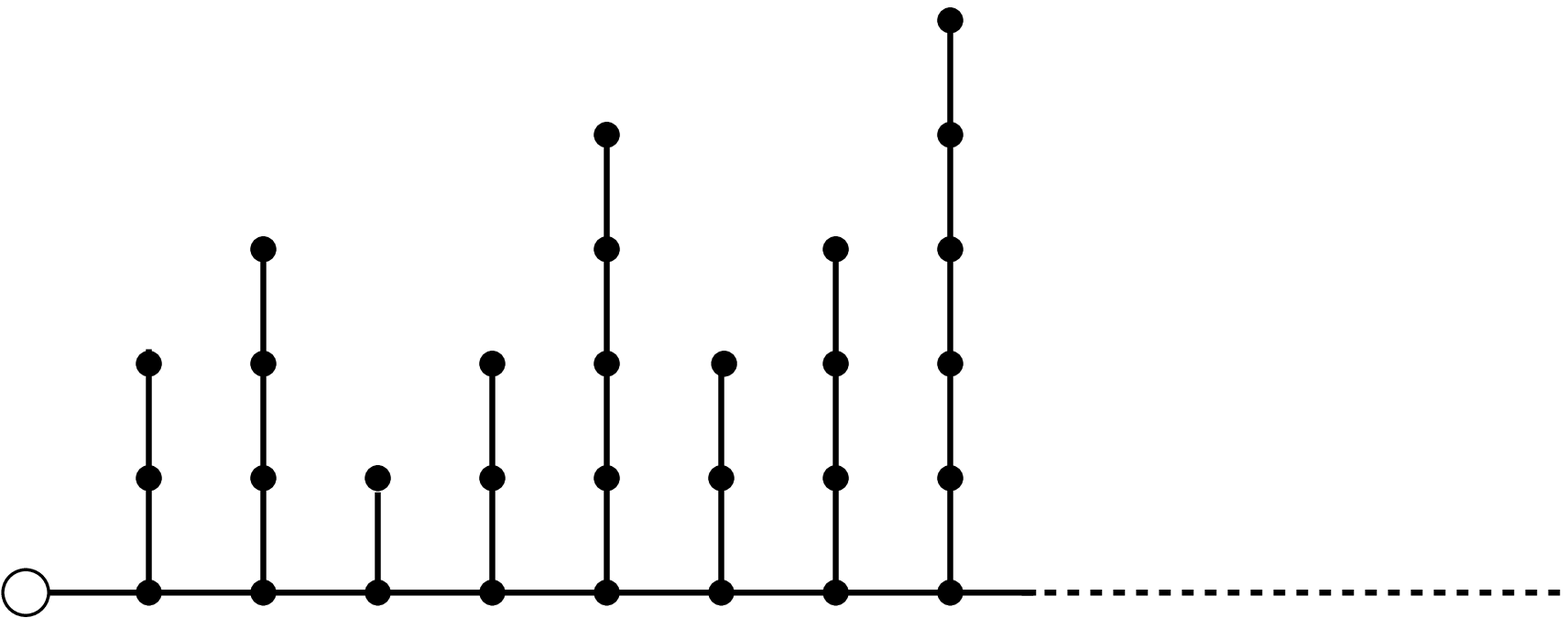}
    \caption{A comb.}
    \label{fig1}
  \end{center}
\end{figure}

Now let $\C{C}$ denote the collection of all combs and define a probability measure $\nu$ on $\C{C}$ by letting the length of the teeth be identically and independently distributed by the measure $\mu$.  We will refer to the set $\C{C}$ equipped with the probability measure $\nu$
as a random comb. Measurable subsets $\C{A}$ of $\C{C}$ are called {\it events} and $\nu (\C{A})$ is the probability of the event $\C{A}$.
The measure of the set of combs $\C{A}$ with teeth at $n_1, n_2,\ldots ,n_k$ having lengths $\ell_1,\ell_2,\ldots ,\ell_k$ is 
\beq
\nu (\C{A})=\prod_{j=1}^k \mu (\ell_j).
\eeq
For any $\nu$-integrable function $F$ defined on $\C{C}$ we define the expectation value
\beq
\langle F(C)\rangle = \int F(C)\,d\nu.
\eeq
We will often use the shorthand $\bar F$ for $\langle F(C)\rangle$.

\subsection{Random Walks}
We consider simple random walk on the comb $C$ and count the time $t$ in integer steps. At each time step the walker moves from its present location  at vertex $v$  to one of the neighbours of $v$ chosen with equal probabilities $\sigma(v)^{-1}$.  Unless otherwise stated  the walker
always starts at the root at time $t=0$.  

The generating function for the probability $p_C(t)$ that the walker is at the        
root at time $t$, having left it  at $t=0$, is defined by
\beq
Q_C(x)=\sum_{t=0}^\infty (1-x)^{t/2}p_C(t)
\eeq
and we denote by $P_C(x)$ the corresponding generating function for the probability that the walker returns to the root for the \emph{first} time, excluding the trivial walk of length 0. Since walks returning to the root   can be decomposed  into walks returning for the 1st, 2nd etc time we have
\beq\label{QPrelation}
Q_C(x) = \frac{1}{1-P_C(x)}.
\eeq
It is convenient to consider contributions to $P_C(x)$ and $Q_C(x)$ from walks which are  restricted. 
Let $P_C^{(n)}(x)$ denote the contribution to $P_C(x)$ from walks whose maximal distance along the spine from the root is $n$ and  define 
\beq P^{(<n)}_C(x)=\sum_{k=0}^{n-1} P_C^{(k)}(x)\eeq
which is the contribution from all walks which do not  reach the point $n$ on the spine. Similarly we define
\beq P^{(>n-1)}_C(x)=\sum_{k=n}^{\infty} P_C^{(k)}(x).\eeq
Clearly $P_C(x)$ can be recovered from $P^{(<n)}_C(x)$ by setting $n\to\infty$. We define the corresponding restricted contributions to $Q_C(x)$ in the same way.
By decomposing walks contributing to  $P^{(<n)}_C(x)$ into a step to $1$, walks returning to $1$ without visiting the root, and finally a step back to the root it is straightforward to show that
\beq\label{Precurrence} P^{(<n)}_C(x)=\frac{1-x}{3-P_{T_1}(x)-P^{(<n-1)}_{C_1}(x)},\eeq
where we have adopted the convention that for the empty tooth, $T=\emptyset$,
\beq P_\emptyset(x)=1.\eeq
The relation \eqref{Precurrence} can be used to compute the generating function explicitly for any comb with a simple periodic structure and we list some standard results in \ref{StandardResults}.

There are a number of elementary lemmas which characterise the dependence of $P_C(x)$ on the length of the teeth and the spacing between them \cite{Durhuus:2005fq}. We state them here in a slightly generalized form which is useful for our subsequent manipulations.
\begin{lemma}
\label{MonoLem1}
The function $P^{(<n)}_C(x)$ is a monotonic increasing function of $P_{T_k}(x)$ and $P^{(<n-k)}_{C_{k}}(x)$ for any $n> k\ge 1$.
\end{lemma}
\begin{lemma}
\label{MonoLem2}
$P^{(<n)}_C(x)$ is a decreasing function of the length, $\ell_k$, of the tooth $T_k$ for any $n>k\ge 1$.
\end{lemma}
\begin{lemma}
\label{RearrangeLem1}
Let $C'$ be the comb obtained from $C$ by swapping the teeth $T_k$ and $T_{k+1}$, $k<n-1$. Then $P^{(<n)}_C(x)>P^{(<n)}_{C'}(x)$ if and only if $ P_{T_k}(x)> P_{T_{k+1}}(x)$.
\end{lemma}
The proofs use \eqref{Precurrence} and follow those given in \cite{Durhuus:2005fq} for the case $n=\infty$.

An important corollary, valid for any comb, of these lemmas is that
\bea \label{Qbounds} x^{-\quarter}\leq Q_C(x)\leq x^{-\half},\eea
which we will refer to as the trivial upper and lower bounds on $Q_C(x)$. The result follows from Lemma \ref{MonoLem2} with $n=\infty$, which gives
\bea
P_*(x) \leq P_C(x)\leq P_\infty(x),
\eea
 and the explicit expressions for  $P_*(x)$ and $P_\infty(x)$ given in \ref{StandardResults}.

\subsection{Two point functions}
Two point correlation functions on the comb correspond to the probability of a walk beginning at the root being at  a particular vertex on the spine at time $t$. In particular, let $p_C(t;n)$ denote the probability that a random walk that starts at the root at time zero is at the vertex $n$ on the spine at time $t$ having not visited the root in the intervening period. We will refer to the generating function for these probabilities as the two point function,  $G_C(x;n)$,  and define it by
\bea
G_C(x;n) = \sum^{\infty}_{t=1} (1-x)^{t/2} p_C(t;n).
\eea
$G_C(x;n)$ may be expressed as 
\bea
G_C(x;n)=\sigma(n)(1-x)^{-n/2}\prod_{k=0}^{n-1}P_{C_k}(x)
\eea
which may be used in conjunction with Lemma  \ref{MonoLem2} 
 to obtain the bounds,
\bea
\frac{G_*(x;n)}{3} \leq \frac{G_C(x;n)}{\sigma(n)}\leq \frac{G_\infty(x;n)}{2}.
\eea
Now let $r_C(t;n)$ denote the probability that a random walk that starts at the root at time zero is at the vertex $n$ on the spine for the first time at time $t$ having not visited the root in the intervening time. We define the modified two point function, $G^{0}_C(x;n)$, by, 
\bea
G^{0}_C(x;n) = \sum^{\infty}_{t=1} (1-x)^{t/2} r_C(t;n)
\eea
and note the following lemmas; 
\begin{lemma}
\label{P2bound}
The contribution $P_C^{(>N)}(x)$ to $P_C(x)$ from  walks whose maximal distance from the root is $N$ or greater satisfies 
\beq
P_C^{(>N-1)}(x) \leq 3x^{-1/2} G_C^0(x;N)^2.    
\eeq
\end{lemma}
The proof is given in section 2.4 of \cite{Durhuus:2005fq}.  
\begin{lemma}
\label{Mod2ptBounds}
The modified two point function satisfies 
\bea
G^{(0)}_{*}(x;n)\leq G^{(0)}_{C}(x;n) \leq G^{(0)}_{\infty}(x;n).
\eea
\end{lemma}
To prove this note that 
\bea G^{(0)}_{C}(x;n)=\frac{(1-x)^{-(n-2)/2}}{\sigma(n-1)}\prod_{k=0}^{n-2}P^{(<n-k)}_{C_k}(x)
\eea 
and use Lemma  \ref{MonoLem2} .

\subsection{Spectral dimension and the continuum limit}\label{dsdiscussion}
The spectral dimension of random combs was studied in \cite{Durhuus:2005fq}. In this work the probability distributions $\mu(\ell)$ for the length $\ell$ of a tooth were chosen to be fixed sets of numbers so the teeth at adjacent sites on the spine are not only independent but generally show large fluctuations relative to each other.   In these circumstances the spectral dimension, $d_s$, describes the large $t$ dependence of the return probability $p_C(t)$
and is given by
\beq d_s=-2 \lim_{t\to\infty} \frac{\log(p_C(t))}{\log t}\label{tdependence}\eeq
if the limit exists.
In fact it is much more convenient to deal with generating functions and by a Tauberian theorem \cite{Flajolet}
we expect that if \eqref{tdependence} holds then, as $x\to 0$,
\beq Q_C(x)\sim x^{-1+d_s/2},\label{QCdsdef}\eeq
where by $f(x)\sim g(x)$ we mean that
\beq c g(x)\le f(x)\le c' g(x),\quad 0<x<x_0,\eeq
where $c$, $c'$ and $x_0$ are positive constants. The property \eqref{QCdsdef} was adopted in \cite{Durhuus:2005fq} as the definition of spectral dimension, assuming it exists. Heuristically the spectral dimension characterizes certain aspects of the long distance structure of a graph as observed by a walker who goes on a very long walk and hence probes that structure. The spectral dimension of an ensemble average is defined in the same way, simply replacing $p_C(t)$ and $Q_C(x)$ by their respective expectation values.

In this paper we study the possibility of different spectral dimension on different distance scales. To do this we have to generalize our definition from \eqref{tdependence} or \eqref{QCdsdef} and introduce at least one characteristic distance scale $L\gg 1$ into the probabilities $\mu(\ell)$ which determine the structure of the comb. We then assign the value $a$ to the distance between adjacent vertices in the graph and take the limit $a\to 0$ and $L\to\infty$ in such a way that the scaled combs have a finite characteristic distance scale; we will refer to this limit as the `continuum' limit and quantities which exist in this limit as continuum quantities. Walks  much longer than $L$ will probe different structure from walks much shorter than $L$ but nonetheless both can be very long in units of the underlying cut-off scale $a$.

In the following sections we will denote  dependence of a function on a number of variables $L_i$,  $i=1,\ldots,N$,  by  $L_i$ passed as one of the function arguments. Given a random comb ensemble specified by $\mu(\ell; L_i)$ and the corresponding $\bar{Q}(x; L_i)$  we define  
\beq\label{Qlimit}
\tilde{Q}(\xi; \lambda_i) = \lim_{a\rightarrow 0} a^{\Delta_\mu} \bar{Q}(a \xi;a^{-\Delta_i}\lambda_i^{\Delta_i}),
\eeq
where the scaling dimensions $\Delta_\mu$ and $\Delta_i$ are chosen to ensure a non-trivial limit and the combinations $\xi\lambda_i$ are dimensionless. $\tilde Q$ can be used to define the spectral dimension at short and long distances.

\newcommand\Qbar{\bar{Q}}
In the following discussion we assume for simplicity that there is just one scale $L$ and that the spectral dimension in the sense of \eqref{tdependence}  exists for $\avg{p_C(t)}$ which implies  that there exists a constant  $t_0$ such that $\avg{p_C(t+1)} <\avg{p_C(t)}, t>t_0$.
Note that 
\bea\label{short} 
&&\fl{\sum_{t=0}^T \avg{p_C(t)} (1-x)^{t/2}= \Qbar(x;L)-\sum_{t=T+1}^\infty \avg{p_C(t)}(1-x)^{t/2}}\nn\\
&&\qquad\qquad \fl{=\Qbar(x;L)-(1-x)^{(T+1)/2}\sum_{t=0}^\infty \avg{p_C(t+T+1)}(1-x)^{t/2}}\nn\\
&&\qquad\qquad \fl{>\Qbar(x;L)-(1-x)^{(T+1)/2}\sum_{t=0}^\infty \avg{p_C(t)}(1-x)^{t/2}.}
\eea
Now choose   
\beq T=\left\lfloor a^{-1}\frac{1}{\xi \log(1+\frac{1}{\xi\lambda})}\right\rfloor-1\eeq
and set  $x=a\xi $ and $L= a^{-\Delta}\lambda^\Delta$ in \eqref{short} to get
\bea \fl{\quad a^{\Delta_\mu} \Qbar(a\xi ;a^{-\Delta}\lambda^\Delta )\left(1-\exp\left(-\xi\lambda\right)\right)< a^{\Delta_\mu}\sum_{t=0}^T \avg{p_C(t)} (1-\xi a)^{t/2}<a^\Delta_\mu\Qbar(a\xi ;a^{-\Delta}\lambda^\Delta ).}\eea
Provided that the limit in \eqref{Qlimit} exists we see that the behaviour of $\tilde Q(\xi;\lambda)$ as $\xi\to\infty$ characterizes the properties of walks of continuum time duration less than 
\beq \lim_{\xi\to\infty}\frac{1}{\xi \log(1+\frac{1}{\xi\lambda})}=\lambda,\eeq
  and we define the  spectral dimension $d^0_s$ at short distances  by
\beq d^0_s= 2\left(1+\lim_{\xi\to\infty} \frac{\log(\tilde Q(\xi;\lambda))}{\log \xi}\right),\label{dsShort}\eeq
provided this limit exists. 

We can define the spectral dimension at long distances in a similar way. First  note that 
by \eqref{Qbounds}
\bea\fl{  \sqrt{T}\ge\sum_{t=0}^\infty \avg{p_C(t)} \left(1-\Tinv\right)^{t/2}>\sum_{t=0}^T \avg{p_C(t)} \left(1-\Tinv\right)^{t/2} >\left(1-\Tinv\right)^T\sum_{t=0}^T \avg{p_C(t)}} \eea
so that 
\bea \Qbar(x;L) -\sqrt{T}(1-\Tinv)^{-T}<\ \sum_{t=T+1}^\infty \avg{p_C(t)}(1-x)^{t/2}< \Qbar(x;L).\eea
This time letting $T=\lfloor a^{-1}\xi^{-1} \log(1+\xi\lambda)\rfloor-1$ we get
\bea\fl{ a^{\Delta_\mu}\Qbar(a\xi;a^{-\Delta}\lambda^\Delta) -e\sqrt{\xi^{-1} \log(1+\xi\lambda)}<  a^{\Delta_\mu}\sum_{t=T+1}^\infty \avg{p_C(t)}(1-\xi a)^{t/2}< a^{\Delta_\mu}\Qbar(a\xi; a^{-\Delta}\lambda^\Delta)}.\nn\\\eea
Provided that the limit in \eqref{Qlimit} exists and that  $\tilde Q(\xi;\lambda)$ diverges as $\xi\to 0$ we see that its behaviour  characterizes the properties of walks of continuum time duration  greater than   $\lim_{\xi\to0} \xi^{-1} \log(1+\xi\lambda)=\lambda$ .  We then define the spectral dimension $d_s^\infty$ at long distances to be
\beq d^\infty_s= 2\left(1+\lim_{\xi\to 0} \frac{\log(\tilde Q(\xi;\lambda))}{\log \xi}\right),\label{dsLong}\eeq
provided this limit exists.

It is by no means obvious that there are graph ensembles for which the limits \eqref{Qlimit} followed by \eqref{dsShort} and \eqref{dsLong}  exist. However in the rest of this paper  we will show for comb ensembles of increasing generality that this is indeed the case. Clearly at the very least any such ensemble must have a characteristic distance scale $\lambda$ that survives the continuum limit otherwise such behaviour is impossible. In all the examples given in this paper it turns out that  the exponent $\Delta_\mu=\half$.  


\section{A simple comb} 
\label{SimpleExample}
We now introduce a random comb whose spectral dimension differs on long and short length scales and thus illustrates that the behaviour described in section \ref{dsdiscussion} can actually occur. 
This comb is defined by the measure,
\bea \label{easycomb}
\mu(\ell;L)&=&\cases{1-\frac{1}{L},&$\ell=0,$\\ \frac{1}{L},& $\ell=\infty,$\\0,&\rm{otherwise.}}
\eea
This random comb  has infinite teeth and they occur with an average separation of $L$. Intuitively we would expect that if a random walker did not move further than a distance of order $L$ from its starting position it would not see the teeth and therefore would measure a spectral dimension of one. If however it were allowed to explore the entire comb it would see something roughly equivalent to a full comb and so feel a much larger spectral dimension. 
To prove this intuition correct we proceed by computing upper and lower bounds for $\bar{Q}$ which are uniform in $L$ and for $0<x<x_0$, where the constant $x_0$ is equal to one unless otherwise stated, and then take the continuum limit to obtain bounds for $\tilde{Q}$ . 


With complete generality we may obtain a lower bound on $\bar{Q}(x)$ by use of Jensen's inequality which takes the form,
\begin{lemma}
\label{GenLBLem}
Let $\bar{P}_T(x;L_i)$ be the mean first return probability generating function of the teeth of the comb defined by $\mu(\ell; L_i)$, then 
\bea
\label{Qeqn}
\bar{Q}(x;L_i) \geq (1+x-\bar{P}_T(x;L_i))^{-\half}.
\eea
\end{lemma}
The proof is given in \cite{Durhuus:2005fq}. For the  comb \eqref{easycomb} we have
\bea
\bar{P}_T(x;L) = 1 - \frac{1}{L}(1-P_{\infty}(x)) = 1 - \frac{\sqrt{x}}{L} 
\eea
 which implies
\bea
\bar{Q}(x;L) \geq \left(\frac{\sqrt{x}}{L}+x\right)^{-\half}.
\eea
Letting $x = a \xi$ and $L = a^{-\half}\lambda^{\half}$ gives

\bea
\label{SimpleLB}
\tilde{Q}(\xi; \lambda)= \lim_{a\to 0} a^\half \Qbar(a \xi; a^{-\half}\lambda^{\half})\geq
\xi^{-\half}\left(\upsilon^{-\half} + 1\right)^{-\half},
\eea
where we have introduced the dimensionless variable $\upsilon = \xi \lambda$.


To find an upper bound on  $\bar{Q}(x;L)$ we follow \cite{Durhuus:2005fq} and use Lemmas \ref{MonoLem1},  \ref{MonoLem2} and  \ref{RearrangeLem1} to compare a typical comb in the ensemble with the comb consisting of a finite number of infinite teeth at regular intervals.  
First we define the event
\bea
\C A(D,k) = \{C:D_i \leq D:i=0,...,k\}.
\eea
where $D_i$ is the distance between the $i$ and $i+1$ teeth and then
 write,
\bea
\label{Qintegral}
\bar{Q}(x;L) &=& \int_{\C{C}} Q_C(x;L) d\nu \nn\\
 &=& \int_{\C{C}/\C{A}(D,k)} Q_C(x;L) d\nu + \int_{\C{A}(D,k)} Q_C(x;L) d\nu.
\eea
Since the $D_i$ are independently distributed
\beq \nu(\C A(D,k))=(1-(1-1/L)^D)^k.\eeq

Consider a comb $C \in \C A(D,k)$; then by Lemmas \ref{MonoLem1},  \ref{MonoLem2} and  \ref{RearrangeLem1}, 
\beq P_C(x;L) \leq P_{C'}(x),\eeq
where $C'$ is the comb obtained by removing all teeth beyond the $k$ tooth and moving the remaining teeth so that the spacing between each is $D$. Now we can write 
\beq P_{C'}(x) = P^{(<Dk)}_{C'}(x) + P^{(>Dk-1)}_{C'}(x).\eeq Since the walks contributing to $P^{(<Dk)}_{C'}(x)$ do not go beyond the last tooth we have \beq P^{(<Dk)}_{C'}(x) \leq P_{*D}(x),\label{bd1}\eeq where $*D$ denotes the comb consisting of infinite teeth regularly spaced and separated by a distance $D$. Using \eqref{bd1}, Lemmas \ref{P2bound} and \ref{Mod2ptBounds} we have,
\beq\label{PCupper1}
P_C(x;L) \leq P_{{*D}}(x) +  3x^{-\half} G^{(0)}_\infty(x; Dk)^2
\eeq
uniformly in $\C A$. $P_{*D}(x)$ and $G^{(0)}_\infty(x; n)$ are given in Appendix A. 
Now set $D = \lfloor\tilde{D}\rfloor$ and $k=\lceil\tilde{k}\rceil$, where,
\beq 
\tilde{D} = 2 L|\log x L^2|, \qquad \tilde{k}= (xL^2)^{-1/2}.
\eeq
Since $G^{(0)}_\infty(x; n)$ is manifestly a monotonic decreasing function of $n$  and $ P_{*D}(x)  $ an increasing function of $D$,
\beq\label{bd100}
\fl{\bar{Q}(x;L) \leq x^{-1/2} (1-(1-(1-1/L)^{\tilde{D}-1})^{\tilde{k}+1}) + Q_{U}(x)(1-(1-1/L)^{\tilde{D}})^{\tilde{k}} }
\eeq
where we have used \eqref{Qbounds} and
\beq Q_U(x) = \left[1- P_{*\tilde{D}}(x) - 3x^{-\half} G^{(0)}_\infty(x; (\tilde{D}-1)\tilde{k})^2\right]^{-1}.
\eeq
Taking the continuum limit of \eqref{bd100} and using the results of \ref{StandardResults} then gives
\beq\label{SimpleUB} \tilde{Q}(\xi;\lambda) \leq \xi^{-1/2} F(\xi\lambda),\eeq
where
\bea F(v)=\cases{1+o(v^{-1}), &$v\to \infty$,\\v^\quarter\sqrt{\abs{\log v^2}}+o(v^\half), &$v\to 0$.}\eea
It follows from \eqref{dsShort}, \eqref{dsLong}, \eqref{SimpleLB} and  \eqref{SimpleUB}
that
\bea d_s^0=1,\qquad
d_s^\infty=\threehalves.\eea

\section{Combs with Power Law Measures}
We now consider slightly more general combs in which the measure on the teeth is a power law of the form,
\bea \label{powerdist}
\mu(\ell;L)&=&\cases{1-\frac{1}{L},&$\ell=0,$\\ \frac{1}{L}C_{\alpha} \ell ^{-\alpha},& $\ell>0,$}
\eea
where $C_\alpha$ is a normalisation constant and as before $L$ plays the role of a distance scale.  We consider laws in the range $2>\alpha>1$ as it is known that for $\alpha\ge 2$ the comb has spectral dimension $d_s=1$ in the sense of \eqref{QCdsdef} \cite{Durhuus:2005fq} and therefore it is not possible to get a spectral dimension deviating from 1 on any  scale.

To compute a lower bound on the return probability generating function for the above distribution we apply Lemma \ref{GenLBLem} and reduce the problem to computing an upper bound on $1-\bar{P}_T(x)$.
The first return generating function $P_\ell(x)$ for a tooth of length $\ell$  is recorded in \eqref{Pell}; bounding  $\tanh(u)$ above by the function $f(u) = u$ for $u<1$ and $f(u) = 1$ for $u\geq 1$ gives\footnote{In this particular case we could in fact compute $1-\bar{P}_T(x)$ exactly by the Abel summation formula. However the bound we use is good enough to give the desired result with the advantage that the calculation can be done with elementary functions.}
\bea
\label{LB1}
1-\bar{P}_T(x;L_i) &\leq&  \sqrt{x} \left[ m_{\infty} (x)\sum^{[m_\infty^{-1}]}_{\ell=1}  \mu(\ell;L_i) \ell + \sum^{\infty}_{\ell=[m_\infty^{-1}]+1} \mu(\ell;L_i)\right] \nn \\
&\leq& \sqrt{x} -m_\infty(x)\sqrt{x} \int^{\frac{1}{m_\infty(x)}}_0\left( \sum^{[u]}_{\ell=0} \mu(\ell;L_i)\right) du.
\eea
To obtain the second inequality we have applied the Abel summation formula. We therefore have,
\begin{lemma}
\label{GenLBLem2}
For a random comb defined by the measure $\mu$,
\beq
1-\bar{P}_T(x;L_i) \leq \sqrt{x} -m_\infty(x)\sqrt{x} \int^{\frac{1}{m_\infty(x)}}_0 \chi(u;L_i) du
\eeq
where the cumulative probability function $\chi(u;L_i)$ is defined by $\chi(u;L_i)  = \sum^{[u]}_{\ell=0} \mu(\ell;L_i)$.
\end{lemma}
We will see shortly that all behaviour of the spectral dimension of the continuum comb is encoded in the asymptotic expansion of ${\chi}(u;L_i)$ as $u$ goes to infinity \footnote{In general it is not obvious that this asymptotic expansion exists due to the discontinuous nature of $\chi$. We will address this issue later when we consider generic measures.}.
In the present case ${\chi}(u;L)$ is trivially related to the partial sum of the Riemann $\zeta$-function whose leading asymptotic behaviour is well known and we find 
\bea
{\chi}(u;L) &=&  1 - \frac{C_\alpha}{L}\frac{u^{1-\alpha}}{\alpha-1} +\delta(u),
\eea
where 
\bea \abs{\delta(u)}<\frac{c}{L}u^{-\alpha}, \quad u\ge 2.\eea
It follows that for $x<x_0$, where $m_\infty(x_0)=\half$,
\bea 1-\bar{P}_T(x;L) \leq  m_\infty(x)\sqrt{x}\left(\frac{b_1}{L}m_\infty(x)^{\alpha-2}+\frac{b_2}{L}m_\infty(x)^{\alpha-1}+\frac{b_3}{L}\right),\eea
with $b_{1,2,3}$ being  constants depending only on $\alpha$ and $b_1>0$.
Choosing  $L= a^{-\Delta'} \lambda^{\Delta'}$ with $\Delta' = 1-\alpha/2$  yields a lower bound on the continuum return generating function,
\beq
\label{PowerLB}
\tilde{Q}(\xi,\lambda)
\geq \xi ^{-1/2}\left(1+b_1 (\xi\lambda)^{-(1-\alpha/2)}\right)^{-1/2}. \eeq


To obtain a comparable upper bound 
we need 
\begin{lemma}
\label{GenUBLem}
For any random comb and positive integers $H$, $D$ and $k$, the return probability generating function is bounded above by
\bea
\label{GeneralUB}
\bar{Q}(x;L_i) \leq x^{-1/2} (1-(1-(1-p)^{{D}})^{{k}}) + Q_{U}(x)(1-(1-p)^{{D}})^{{k}} ,
\eea
where
\bea p &=& \sum^{\infty}_{\ell = H+1} \mu(\ell;L_i),  \nn\\ 
\label{GeneralQU}
Q_U(x) &=& \left[1- P_{{H},*{D}}(x) - 3x^{-\half} G^{(0)}_\infty(x; {D}{k})^2\right]^{-1},
\eea
and $P_{H, *D}$ is the first return probability generating function for the comb with teeth of length $H+1$ equally spaced at intervals of $D$.
\end{lemma}
The proof is  a slight modification of the upper bound argument used in Section \ref{SimpleExample}.
First define a long tooth to be one whose length is greater than $H$; then the probability that a tooth at a particular vertex is long is
\beq p = \sum^{\infty}_{\ell = H+1} \mu(\ell;L_i).  \eeq
Define the event
\bea
\label{Aevent}
\C A(D,k) = \{C:D_i \leq D:i=0,...,k\}
\eea
where now $D_i$ is the distance between the $i$ and $i+1$ long teeth so that
\bea
\bar{Q}(x;L_i) &=& \int_{\C{C}} Q_C(x;L_i) d\nu \nn\\
&=& \int_{\C{C}/\C{A}(D,k)} Q_C(x;L_i) d\nu + \int_{\C{A}(D,k)} Q_C(x;L_i) d\nu.
\eea
Since the $D_i$ are independently distributed
\beq\label{nuD} \nu(\C A(D,k))=(1-(1-p)^D)^k.\eeq
Now use Lemmas \ref{MonoLem2} and \ref{RearrangeLem1} in turn  to note that for 
\beq P_{C\in \C A(D,k)}(x,L) \leq P_{C'}(x,L)\eeq
where $C'$ is the comb in which all teeth but the first $k$ long teeth have been removed and the remaining long teeth have been arranged so that they have length $H$ and a constant inter-tooth distance  $D$. By the same arguments as we used in Section 3 to get \eqref{PCupper1} we obtain the bound 
\beq
\label{PCpexp}
P_{C\in \C A(D,k)}(x,L) \leq P_{{H}, *{D}} +3 x^{-1/2} G^{(0)}_\infty (x,Dk)^2.
\eeq
Lemma \ref{GenUBLem} then follows from \eqref{Qbounds},  \eqref{nuD} and \eqref{PCpexp}. We now specialise to the power law measure \eqref{powerdist} and set $H=\lfloor\tilde{H}\rfloor$, $D = \lfloor\tilde{D}\rfloor$ and $k=\lceil\tilde{k}\rceil$, where
\bea
\label{LambdaK}
\tilde{H} &=& x^{-1/2} \nn\\
\tilde{D} &=& (\Delta'+1) \frac{\alpha-1}{c_\alpha} x^{\Delta'-1/2} L |\log x L^{1/\Delta'}| \\
\tilde{k} &=& (x L^{1/\Delta'})^{-\Delta'}.\nn 
\eea

Using Lemma \ref{GenUBLem}, the scaling expressions for $P_{H, *D}$ and $G_\infty^{(0)}$ given in \eqref{Pelln}, and taking the continuum limit, gives, after a substantial amount of algebra,
\beq\label{PowerUB1} \tilde{Q}(\xi;\lambda) \leq \xi^{-1/2} F(\xi\lambda),\eeq
where
\bea F(v)=\cases{1+O(v^{-1}), &$v\to \infty$,\\c\, v^{1/2-\alpha/4}  \sqrt{\abs{\log v^2}}+O(v^{\Delta'}), &$v\to 0$.}\eea

The main result of this section is  
\begin{theorem}
The comb with the power law measure \eqref{powerdist} for the tooth length has
 \bea d_s^0=1,\qquad
d_s^\infty=2-\frac{\alpha}{2}.\eea
\end{theorem}
The result follows immediately from \eqref{dsShort}, \eqref{dsLong}, \eqref{PowerLB} and  \eqref{PowerUB1}.

\section{Multiple Scales}
Given the results for the power law distribution it is natural to investigate the behaviour for a  random comb that has a hierarchy of length scales. The easiest way to achieve such a comb is through a double power law distribution,
\bea \label{2powerdist}
\mu(\ell;L_i)&=&\cases{1-L_1^{-1}-L_2^{-1},&$\ell=0,$\\ \frac{1}{L_1}C_{1} l ^{-\alpha_1}+\frac{1}{L_2} C_{2} l ^{-\alpha_2} ,& $\ell>0.$}
\eea
We may assume without loss of generality that the length scales $L_i$ scale in the continuum limit to lengths $\lambda_i$ such that $\lambda_1< \lambda_{2}$ and that $1<\alpha_i<2$.  

Following the procedure of previous sections a lower bound on $\tilde Q(\xi;\lambda_i)$ is obtained  by using Lemma \ref{GenLBLem}  and noting that $\chi(x;L_i)$ for this comb is essentially the sum of the cumulative probability functions for each power law. This gives
\bea 
1-\bar{P}_T(x;L_i) \leq  m_\infty(x)\sqrt{x}\sum_{i=1}^2\left(\frac{b_{1i}}{L_i}m_\infty(x)^{\alpha_i-2}+\frac{b_{2i}}{L_i}m_\infty(x)^{\alpha_i-1}+\frac{b_{3i}}{L_i}\right).\nn \\\eea
Choosing $L_i$ to scale like $L_i=a^{-\Delta'_i} \lambda_i^{\Delta'_i}$, where $\Delta'_i = 1-\alpha_i/2$, gives a bound on the continuum return generating function of,
\bea
\label{DoubleLB}
{\xi^{-1/2}\left(c_0 +c_1 (\xi \lambda_1)^{-(1-\alpha_1/2)}+c_2( \xi\lambda_2)^{-(1-\alpha_2/2)} \right)^{-1/2} \leq \tilde{Q}(\xi)}.
\eea
An upper-bound on $\tilde{Q}(\xi)$ is obtained by application of Lemma \ref{GenUBLem} in which we set $H=\lfloor\tilde{D}\rfloor$, $D = \lfloor\tilde{D}\rfloor$ and $k=\lceil\tilde{k}\rceil$, where
\bea
\label{DoubleLambdaK}
\tilde{H} &=& x^{-1/2} \nn\\
\tilde{D} &=& \beta x^{-1/2} G(x L_1^{1/\Delta'_1},x L_2^{1/\Delta'_2})^{-1} |\log x L_1^{1/\Delta'_1}| \\
\tilde{k} &=& G(x L_1^{1/\Delta'_1},x L_2^{1/\Delta'_2}) \nn 
\eea
and for convenience we have introduced the function,
\beq
G(\upsilon_1,\upsilon_2) = \frac{C_1}{\alpha_1-1} \upsilon_1^{-\Delta'_1}+\frac{C_2}{\alpha_2-1} \upsilon_2^{-\Delta'_2} .
\eeq
Using  Lemma \ref{GenUBLem} and the scaling expressions in \ref{StandardResults} then gives
\bea
\label{DoubleUB}
&&\tilde{Q}(x;\lambda_i) \leq \xi^{-1/2}\Bigg[1-(1-\upsilon_1^{-s\beta})^G + \nn \\
&&\frac{ (1-\upsilon_1^{-s\beta})^{G-1} }{3\mathrm{cosech}^2(|\log\upsilon_1^\beta|(1-1/G))-\gamma+\sqrt{\gamma^2+1+2\gamma \coth(|\log\upsilon_1^\beta|/G)}} \Bigg]
\eea
where $\upsilon_i=\xi\lambda_i$,  $s=\sgn(\log\upsilon_1)$, $\gamma=\tanh(1)$ and  we have suppressed the arguments of $G(\upsilon_1,\upsilon_2)$ in order to maintain readability. 

We can now examine  \eqref{DoubleLB} and \eqref{DoubleUB} to see what they tells us about the behaviour of $\tilde{Q}(\xi;\lambda_i)$ on various length scales. 
\begin{itemize}
\item When $\xi\gg\lambda_1^{-1}$ both upper and lower bounds of $\tilde Q(\xi;\lambda_i)$ are dominated by the $\xi^{-\half}$ behaviour so taking the $\xi\to\infty$ limit leads to $d_s^0=1$ as in the previous sections.

\item If $\alpha_1 <  \alpha_2$ then when $\xi\ll\lambda_1^{-1}$ both upper and lower bounds of $\tilde Q(\xi;\lambda_i)$ are dominated by the $\xi^{-\alpha_1/4}$ behaviour so taking the $\xi\to 0$ limit leads to $d_s^\infty=2-\alpha_1/2 $. There is no regime in which $\alpha_2$ controls the behaviour.

\item If $\alpha_2 <  \alpha_1$ then when $\xi\ll\Lambda^{-1}$ where
\beq \Lambda^{-1}=\lambda_1^{(2-\alpha_1)/(\alpha_1-\alpha_2)}\lambda_2^{(2-\alpha_2)/(\alpha_2-\alpha_1)}\eeq
both upper and lower bounds of  $\tilde Q(\xi;\lambda_i)$ are dominated by the $\xi^{-\alpha_2/4}$ behaviour so taking the $\xi\to 0$ limit leads to $d_s^\infty=2-\alpha_2/2 $. However there is an intermediate regime $\Lambda^{-1}\ll\xi\ll\lambda_1^{-1}$ where the $\xi^{-\alpha_1/4}$ behaviour dominates and $\tilde Q(\xi;\lambda_i)$ lies in the envelope given by
\beq c_1 \xi^{-\alpha_1/4} <\tilde Q(\xi;\lambda_i)< c_2\xi^{-\alpha_1/4}\sqrt{\abs{\log\xi^\beta}}, \eeq
where the upper and lower bounds will have corrections suppressed by powers of $\xi \lambda_2$ and the upper bound will also have corrections of order $\xi^\beta$. Both $\lambda_2$ and $\beta$ may be chosen to make the corrections arbitrarily small in this scale range. The system therefore appears to have spectral dimension $\delta_S=2-\alpha_1/2$ in this regime. This is a fairly weak statement because $Q(\xi)$ could in principle exhibit a wide variety of behaviours between its upper and lower bounds; this region is just a part of the crossover regime from $d_S^0$ to $d_S^\infty$. However, as we are free to chose $\lambda_2$ to be as large as we like compared to $\lambda_1$, this regime can exist over a scale range of arbitrarily large size. We therefore can force the leading behaviour of $\tilde Q(\xi;\lambda_i)$ in this range to be as close to a power law with exponent $\delta_S=2-\alpha_1/2$ as we like. This is what might be observed, for example, in a numerical simulation; if the difference between the scales $\lambda_1$ and $\lambda_2$ is large then there will be a substantial range of walk lengths in which the data will indicate a spectral dimension of $\delta_S$. We will refer to a spectral dimension that appears in this weaker way as an {\emph{apparent spectral dimension}} and denote it by $\delta_S$ rather than $d_S$.

\end{itemize}

\section{Generic Distributions}
So far we have considered combs in which the distribution of tooth lengths has been governed by power laws or double power laws. In this section we  extend the results of the previous sections to the case were the form of the tooth length distribution is left arbitrary. The most general situation is that the measure on the combs is a continuous function of some parameters $w_i$. The continuum limit of such a comb is obtained in the usual way but with parameters $w_i$ scaling in a non-trivial way; $w_i = w_{c_i} + a^{d_i} \omega_i$. Given a random comb with such a measure we would like to know how many distinct continuum limits exist and for each compute how the spectral dimension depends on the length scale.

The approach we adopt here closely mimics the arguments of the preceding sections, indeed the main complication is technical. As we have seen the properties of the continuum comb are controlled by the asymptotic expansion of $\chi(u)$ as $u$ goes to infinity. The main difference in the generic case is that we may arrange matters so that the scaling dimensions of the coefficients in the asymptotic expansion are such that sub-leading terms appear in the continuum. For the generic case we obviously have no way of knowing the full asymptotic expansion of $\chi(u; w_i)$. However, we will see that for a large class of measures, the form of the asymptotic expansion is encoded in the asymptotic expansion of a particular generating function for $\mu(\ell; w_i)$.

Our first task is to introduce this generating function and relate it to the asymptotic expansion of $\chi(u; w_i)$. To this end we introduce the notion of a smoothed sum \cite{tao},
\bea
\fl{{\chi_\pm}(u;{w_i}) = \sum^\infty_{\ell=0} \mu(\ell;{w_i}) \eta_\pm(\ell/u)=\mu(0;w_i)+ \sum^\infty_{\ell=1} \mu(\ell;{w_i}) \eta_\pm(\ell/u)} \equiv\mu(0) +\chi_\pm^{(1)}(u;{w_i}), \nn \\
\eea
where $\eta_\pm$ is the smooth cut-off function introduced in \ref{AppendixBump} and $u$ controls where the cut-off occurs. Such smoothed sums are related to ${\chi}(u; w_i)$ by,
\bea
\label{chibound}
{\chi_-}(u;{w_i}) \leq \chi(u;{w_i}) \leq {\chi_+}(u;{w_i}). 
\eea
The reason for introducing the smoothed sums is that we may use powerful techniques from complex analysis to compute their asymptotic expansion (see eg  \cite{Flajolet}). The generating function to which the asymptotic expansion of the smoothed sums is related is  the Dirichlet series generating function of $\mu(\ell;{w_i})$,
\beq
\C{D}_\mu(s;{w_i}) = \sum^{\infty}_{\ell=1} \frac{\mu(\ell;{w_i})}{\ell^s}.
\eeq

\begin{figure}[t]
  \begin{center}
    \includegraphics[width=5cm]{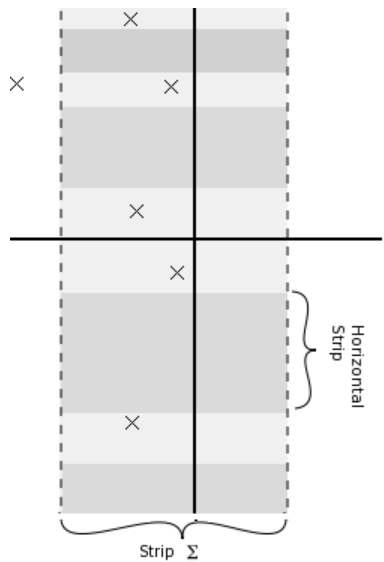}
    \caption{An illustration of the strip $\Sigma$ in which $\C{D}_\mu(s;{w_i})$ satisfies property (1) given given in the text. The poles of $\C{D}_\mu(s;{w_i})$ are denoted by crosses and the horizontal strips in which the growth condition holds are indicated by the dark
grey regions.}
    \label{fig2}
  \end{center}
\end{figure}

We now introduce a number of results and notations,
\begin{itemize}
\item[1.] For a strip in the complex plane $\Sigma(b,a) \equiv \{z: a< \mathrm{Re}[z] < b\}$ where $b>a$, we say $\C{D}_\mu(s;{w_i})$ has {\emph{slow growth}} in $\Sigma(b,a)$ if for all $s \in \Sigma$ we have $\C{D}_\mu(s;{w_i}) \sim O(s^{r})$ for some $r>0$ as $\mathrm{Im}[s] \rightarrow \pm \infty$. We say $\C{D}_\mu(s;{w_i})$ has {\emph {weak slow growth}} if the above property only holds for a countable number of horizontal regions across the strip. See figure \ref{fig2}.
\item[2.] Define $S_\pm$ to be the set containing the triples $(-\sigma+i \tau,-k,-r_\pm)$ such that the Laurent expansion of $\C{D}_\mu(s;{w_i}) \C{M}[\Psi_{\pm\epsilon}](s+1)/s$ about the point $-\sigma+i \tau$ contains the term $r_\pm/(s +\sigma-i \tau)^k$ where $k>0$ and $-\sigma+i \tau \neq0$. Here $\C{M}$ denotes the Mellin transformation and $\Psi_{\pm\epsilon}$ is introduced in \ref{AppendixBump}. We will often find it useful to refer to the elements of $S_\pm$ using an index $j$ and denote the $j$th element of $S_\pm$ by $(-\sigma_j+i \tau_j,-k_j,-r_{j,\pm})$. Note that since $\C{M}[\Psi_{\pm\epsilon}](s+1)$ is analytic in $\Sigma$ then the positions of the poles are determined only by $\C{D}_\mu(s;{w_i})$.

\item[3.] Define the indexing sets $J_R$ and $J_C$, such that if $j \in J_R$ then $\tau_j = 0$ whereas if $j \in J_C$ then $\tau_j \neq 0$ and for both $-\sigma_j+i \tau_j\in \Sigma$ i.e. they index the poles in $\Sigma$ which lie on the real line and off the real line respectively. 
\item[4.] For a Dirichlet series with positive coefficients the abscissa of absolute and conditional convergence coincide. Furthermore, since $\sum^\infty_{l=0} \mu(\ell;{w_i}) = 1$ the abscissa of convergence is less than zero.
\item[5.] Landau's theorem: For a Dirichlet series with positive coefficients there exists a pole at the abscissa of convergence. A corollary of this is that the coefficient of the most singular term in the Laurent expansion about the abscissa of convergence is positive and furthermore it is the right-most pole of the Dirichlet series in the complex plane.
\end{itemize}

If $\C{D}_\mu$ is of slow growth in $\Sigma(a,b)$ and no pole occurs to the right of $\Sigma(a,b)$, then by using the expression \eqref{genasymp} the asymptotic expansion of $\chi_\pm(u;{w_i})$ is, 
\bea
\label{ChiAsymp}
&&{\chi_\pm(u;{w_i}) = \mu(0;{w_i})+\frac{1}{2 \pi i}\oint_C u^{s} \C{D}_\mu(s;{w_i}) \frac{M[\Psi_{\pm\epsilon}](s+1)}{s} ds -R(u;{w_i})} \nn \\
&&= 1 - \sum_{j\in J_R}r_{j,\pm}({w_i})\frac{(\log u)^{k_j-1} }{(k_j-1)!}u^{-\sigma_j}-\\
&&\half \sum_{j\in J_C}\left(r_{j,\pm}({w_i}) u^{i \tau_j}+{r_{j,\pm}}^*({w_i}) u^{-i \tau_j} \right)\frac{(\log u)^{k_j-1}}{(k_j-1)!} u^{-\sigma_j}-R(u;{w_i}), \nn
\eea
where $C$ is the rectangular contour introduced in \ref{AppendixAsymp}\footnote{Since we have assumed there are no poles to the right of $\Sigma$ we may choose $c = b$, where $c$ is the constant appearing in the definition of the contour.}, $r_{j,\pm}$ are the coefficients of the Laurent expansions of $\C{D}_\mu(s;{w_i}) \C{M}[\Psi_{\pm\epsilon}](s+1)/s$, and $R(u)$ is a remainder function which satisfies $R(u) \sim u^{-N}$, where $N>\sigma_j$, as $u$ goes to infinity. Note that the difference between the coefficients $r_{j,+}$ and $r_{j,-}$ is of order $\epsilon$, the parameter introduced in \ref{AppendixBump}, which can be taken to be arbitrarily small and so $r_{j,+}$ and $r_{j,-}$ are for all purposes equal. 

We now are in a position to prove the main result of this section. Note that in the following we require the continuum comb is such that $\Delta_\mu= 1/2$ and so has spectral dimension one on the smallest scales. It seems very likely that this restriction could be lifted
but we will not  pursue such a generalisation here.

\begin{theorem}
\label{mainresult}
For a comb in which the teeth are distributed according to the measure $\mu$ then if $\C{D}_\mu$ has slow growth in the strip $\Sigma = \{z:-1 < Re[z]  \leq 0\}$ and has no poles on its line of convergence besides at the abscissa then continuum limits of the comb exist with $d_s^0=1$ and $d_s^\infty$ taking values in the set $\{\frac{3-\sigma_j}{2}:0<\sigma_j< 1, j\in J_R\}$.\end{theorem}

\begin{proof}
The proof proceeds in much the same way as the proofs in the previous sections; we first derive upper and lower bounds on $\tilde{Q}_C$ and then use these bounds to deduce the behaviour of $\tilde{Q}_C$ in different scale ranges.

We begin with the lower bound which by Jensen's inequality amounts to finding an upper bound on $1-\bar{P}_T$. Note that \eqref{Qeqn} implies that the scaling dimension of $1-\bar{P}_T$ must be greater than or equal to one in order for $\Delta_\mu = 1/2$ and only contributes to the continuum limit if it has scaling dimension one. From Lemma \ref{GenLBLem2} we have,
\bea
\label{LB2}
&&\fl{1-\bar{P}_T(x;{w_i}) \leq  m_\infty(x)\sqrt{x} \int^{\frac{1}{m_\infty(x)}}_0 (1-\chi_-(u;{w_i})) du} \nn \\
&&\qquad\qquad\fl{\leq \sqrt{x}m_\infty(x)C(K;{w_i}) + m_\infty(x)\sqrt{x} \int^{\frac{1}{m_\infty(x)}}_K (1-\chi_-(u;{w_i})) du},
\eea
where $C(K) = \int^K_0 (1-\chi_-(u;{w_i})) du$ is a constant independent of $x$. It is important to recall that since $0\leq \chi_\pm(u;{w_i}) \leq 1$ for all values of $u$ and $w_1,\ldots,w_M$, in particular when $w_1,\ldots,w_M$ assume their critical values, that the scaling dimensions of the coefficients $r_{j,\pm}$ appearing in \eqref{ChiAsymp} must be positive as otherwise $\chi_\pm(u;{w_i})$ would diverge if we were to set $w_1,\ldots,w_M$ to their critical values. Upon performing the integration in \eqref{LB2} we will find that a given $r_{j,\pm}$ now is the coefficient for a number of terms of increasing scaling dimension. Since we require the scaling dimension of $1-\bar{P}_T$ to be greater than or equal to one, only the term with smallest scaling dimension can appear in the continuum limit and we will drop any term that does not appear in the continuum limit.
 If we now substitute \eqref{ChiAsymp} into \eqref{LB2} we find
\bea
\label{PTLB}
&&\fl{1-\bar{P}_T(x;{w_i}) \leq  \sqrt{x}m_\infty(x)C(K;{w_i}) + \sqrt{x}\sum_{j\in J_R}\frac{r_{j,-}({w_i})}{(k_j-1)!}\frac{(-\log m_\infty)^{k_j-1} }{1-\sigma_j}m_\infty^{\sigma_j}+}\nn \\
&&\frac{\sqrt{x}}{2} \sum_{j\in J_C}\left(\frac{r_{j,-}({w_i}) m_\infty^{-i \tau_j}}{1-\sigma_j+i \tau_j}+\frac{{r_{j,-}({w_i})}^*m_\infty^{i \tau_j}}{1-\sigma_j-i\tau_j}  \right)\frac{(-\log m_\infty)^{k_j-1}}{(k_j-1)!}  m_\infty^{\sigma_j},
\eea
where the remainder term $R$ disappears since it cannot appear in the scaling limit. We now suppose we may choose the critical values and scaling dimensions of the parameters $w_i$ such that $r_{j,\pm}$ has a scaling form that can be written as,
\beq
\label{rscaling}
r_{j,\pm} = \left(\frac{a}{\lambda_{j,\pm}}\right)^{\theta_j} \left(-\frac{1}{2}\log a\right)^{{\hat\theta}_j},
\eeq
where $\theta_j$, ${\hat \theta}_j$ and $\lambda_{j,\pm}$ are constants and we have included the factor of $\frac{1}{2}$ for later convenience. From \eqref{PTLB} one can see that since we require $\Delta_\mu = 1/2$ the only terms which appear in the continuum limit are those for which $\theta_j = (1-\sigma_j)/{2}$ and $\hat{\theta}_j = -(k_j-1)$ and so it is useful to define a restricted indexing set
\bea
\tilde{J}_R = \{j \in J_R : \theta_j = (1-\sigma_j)/{2}\quad \mathrm{and} \quad \hat{\theta}_j = -(k_j-1)\} 
\eea
and an equivalent one $\tilde{J}_C$ for $J_C$. The continuum limit is thus,
\bea
&&\fl{1-\bar{P}_T(x;w_i) \leq  a C \xi+\sum_{j\in \tilde{J}_R}\frac{a}{(k_j-1)!(1-\sigma_j)}\xi \left(\xi\lambda_{j,-}\right)^{-\frac{1-\sigma_j}{2}}+} \\
&&\qquad \qquad \fl{\sum_{j\in \tilde{J}_C}\frac{a}{(k_j-1)!}\Bigg(\mathrm{Re}[\phi_j]\cos(\frac{\tau_j}{2} \log (a\xi))+\mathrm{Im}[\phi_j]\sin(\frac{\tau_j}{2} \log (a\xi)) \Bigg)\xi \left(\xi|\lambda_{j,-}|\right)^{-\frac{1-\sigma_j}{2}}}, \nn
\eea
where $\phi_j = (|\lambda_{j,-}|/\lambda_{j,-})^{(1-\sigma_j)/2}(1-\sigma_j+i \tau_j)^{-1}$. The lower bound on $\tilde{Q}_C$ is then obtained from Lemma \ref{GenLBLem}.

Before deriving an upper bound on $\tilde{Q}_C$, we must analyse the consequences of any of the oscillatory terms in the second sum appearing in the continuum limit (i.e. if $\tilde{J}_C$ is non-empty). In the case of the double power law measure we saw that one could obtain intermediate behaviour in which an effective spectral dimension was measured that differed from the UV and IR spectral dimensions. A similar phenomenon occurs in the generic case when the various length scales of the continuum limit are well separated. If we scale the coefficients of the oscillatory terms such that these terms appear in the continuum limit then we are lead to an inconsistency by the following argument, 
\begin{itemize}
\item[1.] Let the term associated with an oscillatory term have index $j_0$. Consider choosing the scaling form of $x$ and $r_{j,-}$ such that $\xi |\lambda_{j_0,-}| \ll 1$ but $\xi |\lambda_{j,-}| \gg 1$ if $|\lambda_{j,-}| \gg |\lambda_{j_0,-}|$.
\item[2.] If we were to take the scaling limit in such a scenario then the term that dominates the size of $1-\bar{P}_T$ is the one associated with the length scale $|\lambda_{j_0,-}|$. This term will be oscillating around a mean of zero and hence must be going negative infinitely often as we approach the continuum.
\item[3.] Since $\bar{P}_T$ is a probability generating function it cannot be negative, hence showing we have an unphysical limit.
\end{itemize}
The cause of this behaviour is that we declared by fiat that the parameters $w_i$ must be scaled 
to give \eqref{rscaling}; however the parameters must also satisfy the constraints $\mu(\ell;{w_i}) > 0$ and $\sum^\infty_{l=0}  \mu(\ell;{w_i}) = 1$ for all values of the parameters and we have not ensured these constraints are compatible with \eqref{rscaling}. It is sufficient for our purposes to understand that the oscillatory terms prevent the continuum limit being taken for certain walk lengths so they are certainly unphysical; we therefore must scale them so that they disappear in the continuum. We therefore have as a lower bound on $Q_C$,
\bea
\tilde{Q}(\xi;\omega_i) \geq \xi^{-1/2}\left(1+\delta_{\Delta{C(K)},0}C+\sum_{j\in \tilde{J}_R}\frac{1}{(k_j-1)!(1-\sigma_j)}\upsilon_{j,-}^{-\frac{1-\sigma_j}{2}}\right)^{-1/2}.\nn\\\label{GeneLB}
\eea

We now consider the upper bound. Without loss of generality we may arrange that the indexing set $J_R$ has the property that if $j_1 < j_2$ then $\lambda_{j_1,\pm} < \lambda_{j_2,\pm}$. Furthermore define
\bea
&&\fl{H(x) = [\tilde{H}(x)] = [x^{-1/2}]}, \nn \\
&&\fl{D(x) = [\tilde{D}(x)] = [\frac{2\beta}{1-\sigma_I} (1-\chi(\tilde{H}(x)))^{-1} |\log( r_{I,+} x^\frac{\sigma_I-1}{2} (-\frac{1}{2}\log x)^{k_I-1})|]}, \\
&&\fl{k(x) = [\tilde{k}(x)] = [x^{-1/2}(1-\chi(\tilde{H}(x)))]}, \nn
\eea
where $I = \inf \tilde{J}_R$, i.e. $\lambda_{I,\pm}$ is the smallest length scale with a spectral dimension differing from one. Some modified versions of the above quantities will also be needed,
\bea
&&\fl{D_\pm(x) = [\tilde{D}_\pm(x)] = [\frac{2\beta}{1-\sigma_I} (1-\chi_\pm(\tilde{H}(x)))^{-1} |\log( r_{I,+} x^\frac{\sigma_I-1}{2} (-\frac{1}{2}\log x)^{k_I-1})|]},\\
&&\fl{k_\pm(x) = [\tilde{k}_\pm(x)] = [x^{-1/2}(1-\chi_\pm(\tilde{H}(x)))]}. \nn
\eea
It is clear that $\tilde{D}_-(x) \leq \tilde{D}(x) \leq \tilde{D}_+(x)$ and $\tilde{k}_+(x) \leq \tilde{k}(x) \leq \tilde{k}_-(x)$ which together with Lemma \ref{GenUBLem} allows us to conclude that the continuum limit is,
\bea
\label{GeneUB}
&&\fl{\tilde{Q}(x;\omega_i) \leq \xi^{-1/2}\Bigg[1-e^{G_- \log(1-\upsilon_I^{-s\beta})} +} \nn \\
&&\qquad\qquad\fl{\frac{e^{(G_+-1) \log(1-\upsilon_I^{-s\beta})}}{3\mathrm{cosech}^2(|\log\upsilon_I^\beta|(1-1/G_+))-\gamma+\sqrt{\gamma^2+1+2\gamma \coth(|\log\upsilon_I^\beta|/G_+)}} \Bigg]}
\eea
where we have defined $G_\pm(\xi,\lambda_1,\ldots) = \sum_{j\in \tilde{J}_R} \upsilon_{j,\pm}^{-\frac{1-\sigma_j}{2}}/(k_j-1)!$ and arranged that no oscillatory terms appear.

We are now in a position to prove Theorem \ref{mainresult} by analysing the behaviour of \eqref{GeneLB} and  \eqref{GeneUB} for various walk lengths. We first note that both the upper and lower bounds are controlled by the relative sizes of the quantities $\upsilon_j^{-1+\sigma_j/2}$. In particular the largest of these quantities, $\Upsilon\equiv\mathrm{sup}\{\upsilon_j^{-1+\sigma_j/2}: j\in\tilde{J}_R\}$, will determine the behaviour in a particular scale range; if we suppose $\Upsilon = \upsilon_j^{-1+\sigma_j/2}$ for some $j$ then the leading contribution to both the upper and lower bound will be proportional to $\xi^{-(1+\sigma_{j})/4}$. Using this we find the following behaviour,
\begin{itemize}
\item On very short scales corresponding to $\xi\gg\lambda_j^{-1}$ for all $j$ then the lower bound of $\tilde Q(\xi)$ is dominated by the $\xi^{-\half}$ behaviour, which together with the trivial upper bound means that taking the $\xi\to\infty$ limit leads to $d_s^0=1$ as in the previous sections.
\item On very long scales corresponding to $\xi\ll\lambda_j^{-1}$ for all $j$ then there will exist a length scale $\Lambda$ such that for $0< \xi < \Lambda^{-1}$, $\Upsilon = \upsilon_{\hat{j}}^{-1+\sigma_{\hat{j}}/2}$ where $\sigma_{\hat{j}} \leq \sigma_{j}$ for all $j$. Explicitly $\Lambda$ will be given by, 
\beq \Lambda=\sup\{\lambda_j^{-(2-\sigma_j)/(\sigma_j-\sigma_{\hat{j}})}\lambda_{\hat{j}}^{(2-\sigma_{\hat{j}})/(\sigma_{j}-\sigma_{\hat{j}})}:j\in \tilde{J}_R\}.\eeq
If we take the limit $\xi \rightarrow 0$ we obtain $d_s^\infty=(3-\sigma_{\hat{j}})/2$.
\end{itemize}

We see that the spectral dimension increases monotonically on successive length scales and the spectral dimension has measured values given by $d_s = (3-\sigma_{\hat{j}(J)})/2$ therefore proving Theorem \ref{mainresult}.
\end{proof}

It is interesting to note that a constraint on the crossover behaviour exists for generic measures much as it did for the case of the double power law measure. In particular on intermediate scales there exists $J \in \tilde{J}_R$ such that $\xi\ll\lambda_j^{-1}$ for $j\leq J$ and $\xi\gg\lambda_j^{-1}$ for $j > J$. Hence  $\upsilon_j^{-1+\sigma_j/2} \ll 1$ if $j>J$ and so $\Upsilon=\mathrm{sup}\{\upsilon_j^{-1+\sigma_j/2}: j\in\tilde{J}_R, j\leq J\}$. This means there will exist a length scale $\Lambda_J$ such that for ${\lambda_{J+1}}^{-1}< \xi < {\Lambda_J}^{-1}$, $\Upsilon = \upsilon_{\hat{j}(J)}^{-1+\sigma_{\hat{j}(J)}/2}$ where we have defined $\hat{j}(J)$ implicitly by $\sigma_{\hat{j}(J)} \leq \sigma_{j}$ for all $j\leq J$. Of course this does not ensure that ${\lambda_{J+1}}^{-1} < {\Lambda_J}^{-1}$, indeed if the length scales $\lambda_j$ are not sufficiently separated this may not be true and we would not have a scale range in which this term dominated. The expression for ${\Lambda_J}$ is,
\beq \Lambda_J=\sup\{\lambda_j^{-(2-\sigma_j)/(\sigma_j-\sigma_{\hat{j}(J)})}\lambda_{\hat{j}(J)}^{(2-\sigma_{\hat{j}(J)})/(\sigma_{j}-\sigma_{\hat{j}(J)})} :j\in \tilde{J}_R, j \leq J\}\eeq
and so we may choose $\lambda_{J+1}$ independent of $\Lambda_J$ thereby allowing the scale range over which this behaviour exists to be arbitrarily large. This would result in an apparent spectral dimension of $\delta_s=(3-\sigma_{\hat{j}(J)})/2$.

Finally, an interesting application of the techniques used to prove Theorem \ref{mainresult} is that they allow one to analyse a wider class of combs than the class for which the results of \cite{Durhuus:2005fq} are valid. In particular it was proven in \cite{Durhuus:2005fq} that for a random comb the spectral dimension is $d_s = (3-\gamma_0)/2$ where $\gamma_0=\sup\{\gamma\geq 0: I_\gamma <\infty \}$ and $I_\gamma =\sum_{\ell =0}^\infty \mu_\ell\, \ell ^{\gamma}$. This was proved subject to the assumption that there exists $d>0$ such that, 
\beq
\label{RCassmp} \sum^{\infty}_{l=[x^{-1/2}]} \mu(\ell) \sim x^d
\eeq
as $x$ goes to zero.

Given the results in this section we see that $-\gamma_0$ may be interpreted as the abscissa of convergence for the Dirichlet series generating function. Furthermore, it is clear from our results that there are distributions where the assumption \eqref{RCassmp} does not hold and that we may use the techniques we have developed to analyse these cases. Recalling that for a random comb we may compute the spectral dimension using the relation \eqref{QCdsdef} then we must perform a similar analysis to that done for the continuum comb but now only scaling $x$ to zero.

Due to Landau's theorem there always exists a pole at the abscissa. Suppose it is of order $k$ and consider $1-\chi_\pm(u)$, 
\bea
&&\fl{1-\chi_\pm(y) = \left[r_{0,\pm} + \sum^\infty_{j=1} \mathrm{Re}[r_{j,\pm}] \cos({\tau_j} \log y) + \mathrm{Im}[r_{j,\pm}] \sin(\tau_j \log y)\right] \frac{(-\log y)^{k-1}}{(k-1)!} y^{\gamma_0}} \nn \\
&&\qquad\qquad\fl{\equiv \Omega(y; \{r_{j,\pm}\}) \frac{(-\log y)^{k-1}}{(k-1)!} y^{\gamma_0}} 
\eea
where $j$ runs over the poles on the line of convergence and $\Omega(y; \{r_{j,\pm}\})$ is implicitly defined in the second line. We have not included any poles of order less than $k$ since these will not be leading order as $y$ goes to zero and we have assumed the there are no poles of order greater than $k$ as then the above quantity would go negative due to the oscillating terms. Applying the same argument as we did for the continuum comb we obtain the lower bound,
\beq
\fl{\bar{Q}(x) \geq \left[ \frac{c'(k-1)!}{\Omega(x^{1/2}; r_{j,-}(1-\gamma_0+i \tau_j)^{-1})}\right]^{1/2}\left(\half|\log x|\right)^{(1-k)/2} x^{-\frac{\gamma_0+1}{4}}.}
\eeq
where $c'$ is a constant. To obtain an upper bound we apply a very similar argument as used for the continuum comb; the only difference being that we choose,
\bea
H &=& [\tilde{H}] = \beta x^{-1/2} |\log x|, \nn \\
D &=& [\tilde{D}] = \beta (1-\chi(\tilde{H}))^{-1} |\log x|, \nn \\
k &=& [\tilde{k}] = x^{-1/2} (1-\chi(\tilde{H})), \\
D_\pm &=& [\tilde{D}_\pm] = \beta (1-\chi_\pm(\tilde{H}))^{-1} |\log x|, \nn \\
k_\pm &=& [\tilde{k}_\pm] = x^{-1/2} (1-\chi_\pm(\tilde{H})) \nn
\eea
and compute the behaviour as $x$ goes to zero of $G^{(0)}_\infty(x,(\tilde{D}-1)(\tilde{k}-1))$, $P_{\tilde{H}, *\tilde{D}_+}$ and the size of the set $\C A$ in \eqref{Aevent}. The result is,
\beq
\fl{\bar{Q}(x) \leq  \left[\frac{c''(k-1)!}{\Omega(\frac{x^{1/2}}{|\log x^\beta|}; r_{j,+})} \right]^{1/2} \left(|\log\left[ \frac{x^{1/2}}{|\log x^\beta|}\right]|\right)^{(1-k)/2} |\log x^\beta|^{\gamma_0/2}x^{-\frac{\gamma_0+1}{4}}.}
\eeq
where $c''$ is a constant. We see that we reproduce the result of \cite{Durhuus:2005fq} when $k=1$ and there are no poles on the line of convergence. If $k\neq 1$, the spectral dimension is the same as the $k=1$ case as only logarithmic corrections are introduced. If we allow poles to appear on the line of convergence then $\Omega(x; r_{j,\pm}) $ will have an oscillating $x$ dependence which, if the bounds above are tight enough, would also imply that the functional form of $\bar{Q}$ is not a power law and so the spectral dimension does not exist. Whether the above bounds are tight enough to make this conclusion we leave for future work. 

\section{Conclusions}
We have demonstrated that there exist models in which a 
scale dependent spectral dimension can be shown to exist analytically. That we could do this was important as up to now all the evidence for scale dependent spectral dimension in CDT has been numerical in nature and therefore open to the criticism that it might be due to discretisation effects or other numerical artefacts. Furthermore, numerical results do not provide any understanding of the mechanism causing the reduction in dimensionality and so we hope the work begun in this paper may be extended to shed some light on this question. This is not an unreasonable expectation as the techniques used to compute the spectral dimensions of random combs \cite{Durhuus:2005fq} have been extended to allow the computation of the spectral dimensions of random trees \cite{Durhuus:2006vk}. Such random trees are closely related to two-dimensional CDT via the bijection described in \cite{Durhuus:2009sm} where it was used to bound the spectral dimension of the spacetime arising in such models.

One might ask how closely related a random comb is to an actual model of quantum gravity. In models of pure quantum gravity in two dimensions the only observable is the spatial volume of the universe. For the case of both CDT and random trees one may show that the evolution of the spatial volume with time is generated by a Hamiltonian. A crucial difference for the random combs is that the growth of a comb is not a Markovian process in the same sense as the growth of a CDT; the vertices at height $h$ from the root are not related to those at height $h-1$ by a local transfer matrix so a Hamiltonian formulation of the evolution of a comb-like universe is not possible. An extension of the work described in this paper, using the techniques of \cite{Durhuus:2006vk}, to the case of random trees would be very interesting as this would constitute an example of  dimensional reduction for a model which does admit a Hamiltonian formulation.

We have given in Theorem \ref{mainresult} a reasonably complete classification of what behaviours one can find on a continuum comb and indeed it is likely that the cases not covered by the theorem do not have a well defined spectral dimension. 
The behaviour for the combs covered by Theorem \ref{mainresult} is fairly rich as there exists a cross-over region between the UV and IR behaviours in which a hierarchy of apparent spectral dimensions exist. 
The proof of the bijection between two-dimensional CDT and trees in \cite{Durhuus:2009sm} shows that any ensemble of critical  Galton-Watson trees is in bijection with a CDT-like theory. It is only when the random tree has the uniform measure that we obtain precisely CDT on the other side of the bijection. This opens the intriguing possibility of constructing a 
variation of CDT, constructed from an ensemble of random trees with a measure that differs from the uniform measure and with a dependence on a length scale. This length scale could then be scaled while taking a continuum limit, as we have done in the work here. Such a model would likely have intermediate length scales with apparent spectral dimensions different from the UV and IR values. We hope to pursue these possibilities in future work.

\ack{
MA would like to acknowledge the support of STFC studentship PfPA/S/S/2006/04507. GG would like to acknowledge the support of the A.S. Onassis Public Benefit Foundation and the A.G. Leventis Foundation grant F-ZG 097/ 2010-2011. This work is supported by EPSRC grant EP/I01263X/1.}


\appendix

\section{Standard generating functions}
\label{StandardResults}
We record here a number of standard results for generating functions for random walks on combs; the details of their calculation are given in \cite{Durhuus:2005fq}.

On the empty comb $C=\infty$ we have
\bea P_\infty(x)&=&1-\sqrt{x},\label{Pinf}\\
P_\infty^{<n}(x)&=&(1-x)\frac{(1+\sqrt{x})^{n-1}-(1-\sqrt{x})^{n-1}}{(1+\sqrt{x})^{n}-(1-\sqrt{x})^{n}},\label{PinfL}\\
G_\infty^{(0)}(x;n)&=&(1-x)^{n/2}\frac{2\sqrt{x}}{(1+\sqrt{x})^{n}-(1-\sqrt{x})^{n}}.  \label{GinfL}
\eea
Note that we can promote $n$ to being a continuous positive semi-definite real variable in these expressions; $G_\infty^{(0)}(x;n)$ is then a strictly decreasing function of $n$ and $P_\infty^{<n}(x)$ a strictly increasing function of $n$.
The finite line segment of length $\ell$ has
\bea P_\ell(x)&=& 1-\sqrt{x}\frac{(1+\sqrt{x})^{\ell}-(1-\sqrt{x})^{\ell}}{(1+\sqrt{x})^{\ell}+(1-\sqrt{x})^{\ell}}
\eea
which it is sometimes convenient to write as 
\bea P_\ell(x)
&=&\sqrt{x}\tanh\left(m_\infty(x)\ell\right)\label{Pell}
\eea
where 
\beq m_\infty(x)=\half\log\frac{1+\sqrt{x}}{1-\sqrt{x}}.\eeq
Again, $\ell$ can be promoted to a continuous positive semi-definite real variable of which $P_\ell(x)$ is a strictly increasing function. The first return probability generating function for the comb with teeth of length $\ell+1$ equally spaced at intervals of $n$ is given by 
\bea
P_{\ell, *n} (x)= \frac{3-P_{\ell}(x)}{2} - \half \left[\left(3-P_{\ell}(x)-2P_\infty^{<n}(x)\right)^2-4G_\infty^{(0)}(x;n)^2\right]^\half.\label{Pelln}
\eea
and $P_{*n}(x)$ is obtained by setting $\ell=\infty$ in this formula. $P_{\ell, *n} (x)$ is a strictly decreasing function of $\ell$ and increasing function of $n$, viewed as continuous positive semi-definite real variables.

We also need the scaling limits of some of these quantities. They are
\bea \lim_{a\to 0} a^{-\half}G_\infty^{(0)}(a\xi;a^{-\half}\nu)&=&\xi^\half \mathrm{cosech}( \nu\xi^\half)\eea
and
\bea \fl{\lim_{a\to 0} a^{-\half}\left(1-P_{(a^{-\half}\rho), *(a^{-\half}\nu)} (a\xi)\right)=-\half\xi^\half\tanh(\rho\xi^\half)}\nn\\ \qquad\qquad\qquad \qquad+\half\xi^\half\left[4+4\tanh\rho\xi^\half\coth\nu\xi^\half+\tanh^2\rho\xi^\half\right]^\half.
\eea

\section{Bump functions}
\label{AppendixBump}
A function $\psi : \BB{R} \rightarrow \BB{R}$ is a bump function if $\psi$ is smooth and has compact support. We will now prove some properties concerning the Mellin transformation of a bump function $\psi$, $\C{M}[\psi](s)$, which has support on $[a,b]$ where $b>a>0$.

\begin{lemma}
\label{bump1}
The critical strip of the Mellin transform of the $n$th derivative of $\psi$, $\psi^{(n)}$, is $\BB{C}$ for all $n$.
\end{lemma}
\begin{proof}
Recall that the Mellin transform is defined by, $\C{M}[\psi](s) = \int^{\infty}_0 \psi(x) x^{s-1} dx$. Since $\psi$ has compact support we have,
\beq
\C{M}[\psi^{(n)}](s) = \int^{b}_a \psi^{(n)}(x) x^{s-1} dx
\eeq 
and since $\psi$ is smooth, $|\psi^{(n)}|$ is bounded on $[a,b]$ by some constant $K$, so,
\beq
|\C{M}[\psi^{(n)}](s)| \leq K \int^{b}_a x^{s-1} dx
\eeq
and the RHS is finite for all $s \in \BB{C}$ since $b>a>0$. This also shows that $\C{M}[\psi^{(n)}](s)$ is holomorphic for all $s$.
\end{proof}
\begin{lemma}
\label{bump2}
Given $n \in \BB{Z}^+$, $|\C{M}[\psi](\sigma+i D)| \leq \frac{1}{D^n} \C{M}[|\psi^{(n)}|](\sigma+n)$ for all $s \in \BB{C}$.
\end{lemma}
\begin{proof}
Recall from the previous lemma that the critical strip of the Mellin transform of $\psi$ and its derivatives coincides with $\BB{C}$. We can therefore use the integral representation of the Mellin transform to prove statements valid for all $s \in \BB{C}$. By integration by parts,
\beq
\C{M}[\psi^{(n)}](s) = -\int^{b}_a \psi^{(n+1)}(x) \frac{x^{s}}{s} dx = -\frac{1}{s}\C{M}[\psi^{(n+1)}](s+1)
\eeq
and therefore,
\beq
\fl{\C{M}[\psi](s) = \frac{(-1)^n}{\prod^{n-1}_{k = 0} (s+k)} \int^{b}_a \psi^{(n)}(x) x^{s+n} dx =  \frac{(-1)^n}{\prod^{n-1}_{k = 0} (s+k)}\C{M}[\psi^{(n)}](s+n).}
\eeq
Hence,
\beq
|\C{M}[\psi](\sigma+iT)| \leq \frac{1}{D^n} \int^{b}_a |\psi^{(n)}(x)| x^{\sigma + n} dx = \frac{1}{D^n} \C{M}[|\psi^{(n)}|](\sigma+n).
\eeq
\end{proof}
Given a bump function $\Psi_{\pm\epsilon}$ which is always positive, has support $[1,1\pm\epsilon]$ and is scaled such that its integral is one, we define the cut-off function to be,
\beq
\eta_\pm(x) = 1 - {\int^x_{-\infty} \Psi_{\pm\epsilon}(x) dx }.
\eeq
\begin{lemma}
The critical strip of $\C{M}[\eta_\pm](s)$ is given be $Re[s]> 0$. The analytic continuation of $\C{M}[\eta_\pm](s)$ to all $s$ is given by
\beq
\C{M}[\eta_\pm](s) = \frac{1}{s}\C{M}[\Psi_{\pm\epsilon}](s+1).
\eeq
\end{lemma}
\begin{proof}
The analytic continuation of $\C{M}[\eta_\pm](s)$ may be obtained by applying integration by parts to the Mellin transform of $\eta_\pm$ and recalling by lemma \eqref{bump1} that $\C{M}[\Psi_{\pm\epsilon}](s)$ is holomorphic everywhere.
\end{proof}

\section{Asymptotic Series and Dirichlet Series}
\label{AppendixAsymp}
Starting with 
\beq
{S_\pm}(y) 
= \sum^\infty_{\ell=0} \mu(\ell) \eta_\pm(\ell y)=\mu(0)+ \sum^\infty_{\ell=1} \mu(\ell) \eta_\pm(\ell y)\equiv \mu(0) + S_\pm^{(1)}(y)
\eeq
where $\eta_\pm$ is the smooth cut-off function introduced in \ref{AppendixBump} and $y$ controls where the cut-off occurs we take the Mellin transform of $S_\pm^{(1)}(y)$ with respect to $y$,
\bea
\C M[S^{(1)}_\pm](s) &=& \int^\infty_0 S^{(1)}_\pm(y) y^{s-1} dy \nn \\
&=& \C{D}_\mu(s) \C M[\eta_\pm](s) \nn \\
&=& \C{D}_\mu(s) \C M[\Psi_{\pm\epsilon}](s+1)/s
\eea
where $\C{D}_\mu(s)$ is the Dirichlet series associated to the measure,
\beq
\C{D}_\mu(s) = \sum^{\infty}_{\ell=1} \frac{\mu(\ell)}{\ell^s}.
\eeq
It is easy to see that the fundamental strip of $S^{(1)}_\pm(y)$ is $Re[s]>0$, due to the compact support of $\eta_\pm$ on the positive real axis and so the Mellin transform does indeed exist. We may now invert the Mellin transform to obtain,
\beq
S_\pm(y) = \mu(0)+\frac{1}{2 \pi i}\int^{c+i\infty}_{c-i\infty} y^{-s} \C{D}_\mu(s) \frac{\C M[\Psi_{\pm\epsilon}](s+1)}{s} ds.
\eeq
This may be computed by rewriting the above as,
\beq
S_\pm(y) = \mu(0)+\frac{1}{2 \pi i}\oint_C y^{-s} \C{D}_\mu(s) \frac{\C M[\Psi_{\pm\epsilon}](s+1)}{s} ds -R(y)
\eeq
where the contour $C$ is the rectangle composed of the points $\{c-i\infty,c+i\infty,-N + i\infty,-N-i\infty\}$ with $N>0$, $c$ such the contour is the right of all poles and  
\beq
R(y) = \int^{-N+i\infty}_{-N-i\infty} y^{-s} \C{D}_\mu(s) \frac{\C M[\Psi_{\pm\epsilon}](s+1)}{s} ds.
\eeq
If $\C{D}_\mu$ has slow growth in the strip $-N \leq Re[z]  \leq 0$ then due to lemma \eqref{bump2}, which shows that $\C M[\Psi_{\pm\epsilon}]$ decays faster than any polynomial as $t$ goes to infinity, the contributions from integrating along the contours $c+i\infty$ to $-N + i\infty$ and from $c-i\infty$ to $-N - i\infty$ are zero. Furthermore the remainder term $R(y)$ satisfies,
\beq
|R(y)| \leq \frac{y^{N}}{N} \int^{-N+i\infty}_{-N-i\infty} |\C{D}_\mu(s) ||{\C M[\Psi_{\pm\epsilon}](s+1)}| ds
\eeq
and so will only contribute terms of order $y^N$ to $S_\pm(y)$. We therefore have,
\bea
&&S_\pm(y) = \mu(0)+\frac{1}{2 \pi i}\oint_C y^{-s} \C{D}_\mu(s) \frac{\C M[\Psi_{\pm\epsilon}](s+1)}{s} ds\\
&&= 1+\sum_{s_i \in S \cap \Sigma} \res\left[\C{D}_\mu(s) \frac{\C M[\Psi_{\pm\epsilon}](s+1)}{s}y^{-s};s=s_i\right]-R(y)
\eea
where $S$ is the set of positions of the poles of $\C{D}_\mu$ and we have used the fact that $\C{D}_\mu(0) = 1-\mu(0)$. Finally, define ${\chi_\pm}(u) = \sum^\infty_{l=0} \mu(\ell) \eta_\pm(l/u)$. By relating this function to $S_\pm(y)$ we may write,
\bea
\label{genasymp}
&&\chi_\pm(u) = \mu(0)+\frac{1}{2 \pi i}\oint_C u^{s} \C{D}_\mu(s) \frac{\C M[\Psi_{\pm\epsilon}](s+1)}{s} ds\\
&&= 1+\sum_{s_i \in S \cap \Sigma} \res\left[\C{D}_\mu(s) \frac{\C M[\Psi_{\pm\epsilon}](s+1)}{s}u^{s};s=s_i\right]-R(y).
\eea

\end{document}